\documentclass[twoside,leqno]{article}

\usepackage{amsmath,scmjlt2e,latexsym}
\usepackage{amssymb}
\usepackage{amsthm}
\usepackage{bm}
\usepackage{stmaryrd}
   
\def\Z{{\mathbb Z}_2}
\def\ZZ{{\mathbb Z}_2 \otimes {\mathbb Z}_2}
\def\hf{\frac{1}{2}}
\def\g{\mathfrak{g}}
\def\ket#1{\left| #1 \right\rangle}
\def\tF{\tilde{F}}

 \newtheorem{thm}{Theorem}[section]
 \newtheorem{cor}[thm]{Corollary}
 \newtheorem{lem}[thm]{Lemma}
 \newtheorem{prop}[thm]{Proposition}
 \theoremstyle{definition}
 \newtheorem{defn}[thm]{Definition}
 \theoremstyle{remark}

 \numberwithin{equation}{section}

\begin{document}
      

\title{ Verma Modules over a $\ZZ$ Graded Superalgebra and Invariant Differential Equations}

\author{ Naruhiko Aizawa }

      \date{}



\maketitle

\begin{abstract}
Lowest weight representations of the $\ZZ$ graded superalgebra introduced by Rittenberg and Wyler are investigated. 
We give a explicit construction of Verma modules over the $\ZZ$ graded superalgebra and show their reducibility by using  singular vectors. The explicit formula of singular vectors are given and are used to derive partial differential equations invariant 
under the color supergroup generated by the $\ZZ$ graded superalgebra. 
\end{abstract}

\section{Introduction.}
  
The present work aims to study  representations of a $\ZZ$ graded Lie superalgebra (also called color superalgebra) and its application to differential equations. We focus on a simple example of color superalgebras for the sake of simplicity and  investigate Verma modules over it.  The result is used to derive the partial differential equations which are invariant under the $\ZZ$ graded Lie supergroup generated by the color superalgebra.   
   
   Color superalgebras are a generalization of Lie superalgebras introduced by Rittenberg and Wyler \cite{rw1,rw2} (see also \cite{sch,GrJa}). 
The idea of generalization is to extend the $\Z$ graded structure of the underlying vector space of Lie superalgebra to more general abelian groups.  
The group $ \ZZ, $ discussed here, is the simplest non-trivial example of the generalization.  
During the last four decades, structure theory, classification of possible color superalgebras of given dimension \textit{etc.} have been studied by many authors. 
See, for example, \cite{sch2,sch3,Sil,ChSiVO,CART,SigSil,AiSe,StoVDJ,Naru} and references therein. However, applications of color superalgebras to mathematical  and physical problems are very limited \cite{lr,vas,jyw,zhe,Toro1,Toro2,tol,tol2,aktt1,aktt2}.   
In \cite{aktt1,aktt2} it is shown that symmetries of a first order linear partial differential equation, called the L\'evy-Leblond equation \cite{LLE}, are generated by a $\ZZ$ graded superalgebra. 
The authors seek symmetry operators of the  L\'evy-Leblond equation systematically and found that they do not close as a Lie algebra or superalgebra but do in a $\ZZ$ graded superalgebra. Here, we reverse the argument. We start with a simple example of $\ZZ$ graded superalgebra and then study its lowest  weight representations and their irreducibility. Reducibility of the representations is detected by the existence of singular vectors. Explicit formulae of the singular vectors allows us to write down invariant partial differential equations. This is a generalization of the method developed for semi-simple Lie groups \cite{Dob} to  the $\ZZ$ setting. Thus we shall see the method is valid beyond the semi-simple Lie groups. 
 
The plan of this paper is as follows. In the next section 
we give a definition of $\ZZ$ graded superalgebra and present the one, denoted by $\g,$ investigated in this work. 
Lowest weight Verma modules over $\g$ are studied in \S \ref{SEC:VMVM}. 
After discussing the failure of the naive approach, we construct Verma modules over $\g$ by using a trick. 
A list of all singular vectors in the Verma modules is presented with an explicit formulae of the singular vectors. 
The formulae of the singular vectors are used to find differential equations whose symmetries are generated by $\g$ in \S \ref{SEC:PDE}. Summary and some remarks are given in \S \ref{SEC:CR}. 
      

\section{Definition of $\ZZ$ graded superalgebra.}

Let $\g$ be a vector space over $ \mathbb{C} $ (or $ \mathbb{R} $) which is a direct sum of four subspaces labelled by an element of the group $ \ZZ:$
\begin{equation}
 \g =\g_{(0,0)} \oplus \g_{(0,1)} \oplus \g_{(1,0)} \oplus \g_{(1,1)}.
\end{equation}
An inner product of two $\ZZ$ vectors $  \bm{\alpha} = (\alpha_1, \alpha_2), \; \bm{\beta} = (\beta_1,\beta_2)$ is defined as usual:
\begin{equation}
  \bm{\alpha} \cdot \bm{\beta} =  \alpha_1 \beta_1 + \alpha_2 \beta_2.
\end{equation}
Now we give a definition of $\ZZ$ graded superalgebra according to \cite{rw1,rw2}. 
\begin{defn} \label{DEF:CSA}
 If $\g$ admits a bilinear form $ \llbracket \ , \ \rrbracket : \g \times \g \to \g $ satisfying the following three relations, then $\g$ is called a $ \ZZ$ graded color superalgebra:
  \begin{enumerate}
     \item $ \llbracket \g_{\bm{\alpha}}, \g_{\bm{\beta}} \rrbracket \subseteq \g_{\bm{\alpha}+\bm{\beta}}, $
     \item $ \llbracket X_{\bm{\alpha}}, X_{\bm{\beta}} \rrbracket 
            = -(-1)^{\bm{\alpha}\cdot\bm{\beta}} \,
             \llbracket X_{\bm{\beta}}, X_{\bm{\alpha}} \rrbracket, $ 
     \item $ \llbracket X_{\bm{\alpha}}, \llbracket X_{\bm{\beta}}, X_{\bm{\gamma}} \rrbracket \rrbracket (-1)^{\bm{\alpha}\cdot\bm{\gamma}}
            + \text{cyclic perm.} = 0,
            $
  \end{enumerate}
  where $ X_{\bm{\alpha}} \in \g_{\bm{\alpha}} $ and the third relation is called the graded Jacobi identity. 
\end{defn}
When the inner product $ \bm{\alpha} \cdot \bm{\beta} $ is an even integer the graded Lie bracket $ \llbracket \ , \ \rrbracket $ is understood as a commutator, while it is an anticommutator if  $ \bm{\alpha} \cdot \bm{\beta} $ is an odd integer.  
If $N=1,$ then Definition \ref{DEF:CSA} is identical to the definition of Lie superalgebras so that the color superalgebra is a natural generalization of Lie superalgebra. 

The $\ZZ$ graded superalgebra investigated in this work is eight dimensional and its basis and grading are given as follows:
\begin{equation} \label{BFbasis}
   \begin{array}{ccl}
      (0,0) & : & A_-,\ A_+,\ N \\
      (1,0) & : & b_-, \ b_+ \\
      (0,1) & : & a_-, \ a_+ \\
      (1,1) & : & F
   \end{array}
\end{equation}
The basis satisfy the relations (only non-vanishing ones are presented):
\begin{align} \label{BFdefrel}
   [A_-, A_+] &= 4N, & [A_-,N] &= 2A_-, & [A_+, N] &= -2A_+,
  \nonumber \\
   [A_-, b_+] &= 2 b_-,  & [A_+, b_-] &= -2b_+, & [N, b_-] &=-b_-, & [N, b_+] &=b_+,
  \nonumber \\
   [A_-, a_+] &= 2a_-, &  [A_+,a_-] &=-2a_+, & [N,a_-] &=-a_-, & [N, a_+] &=a_+,
  \nonumber \\
   \{ b_-, b_- \} &= 2 A_-, & \{b_-, b_+\} &= 2N, & \{b_+, b_+\} &= 2 A_+,
  \nonumber \\
   [b_-, a_+] &= F,  & [b_+,a_-] &= -F,
  & 
   \{ b_-, F \} &= 2 a_-,  & \{b_+, F\} &= 2 a_+,
  \nonumber \\
   \{ a_-, a_-\} &= 2A_-, & \{a_-,a_+\} &=2N, & \{ a_+, a_+ \} &= 2 A_+,
  \nonumber \\
   \{a_-, F\} &= 2b_-, & \{a_+, F\} &= 2 b_+.
\end{align}
One may see from the defining relations that the color superalgebra has two $osp(1,2)$ subalgebras, $ \langle \; A_{\pm},\; N,\; a_{\pm} \; \rangle $ and $ \langle \; A_{\pm},\; N,\; b_{\pm} \; \rangle. $
This color superalgebra was given in \cite{rw1} as one of the non-trivial examples of $\ZZ$ graded superalgebras. It is also discussed in \cite{Naru} as an $\ZZ$ extension of the superalgebra generated by boson and fermion operators. 
We denote the color superalgebra \eqref{BFbasis} by $\g.$ 

\begin{lem}
The superalgebra $\g$ admits an algebra anti-involution $ \omega : \g \to \g $ defined by
\begin{equation} \label{antiinv}
   \omega(X_{\pm}) = X_{\mp}, \qquad \omega(Y) = Y, \quad X_{\pm} = A_{\pm},\; a_{\pm}, \; b_{\pm}, \quad Y = N,\; F 
\end{equation}
\end{lem}
\begin{proof}
It is checked by straightforward computation.
\end{proof}

Next we give a definition of the $ \ZZ $ graded Grassmann numbers which was also introduced in \cite{rw2}. 
\begin{defn}
 Let $ \bm{\alpha} \in \ZZ $ and $ \zeta_{\bm{\alpha},i} $ be a basis of a vector space over the same field as $\g.$ If the basis satisfies the relations
 \begin{equation}
  \llbracket \zeta_{\bm{\alpha}, i}, \zeta_{\bm{\beta}, j}  \rrbracket = 0,
\end{equation}
then we call $ \zeta_{\bm{\alpha},i} $ the $\ZZ$ graded Grassmann numbers. 
\end{defn}
An extension of an integral and a derivative of the ordinary Grassmann numbers to $\ZZ $ setting is discussed in \cite{AiSe}. 
We use the $\ZZ$ graded Grassmann numbers in the subsequent sections. 
The basis $\zeta_{\bm{\alpha},i}$ may be realized in terms of the ordinary Grassmann numbers and the Clifford algebra \cite{Naru}.  
Let us recall that the Clifford algebra $Cl(p,q)$ is a unital algebra generated by $ \gamma_i\; (i=1,2,\dots, N=p+q)$ subject to the relations:
\begin{equation}
  \{ \gamma_i, \gamma_j \} = 2 \eta_{ij}, \qquad 
  \eta = \text{diag}(\underbrace{+1, \dots, +1}_{p}, \underbrace{-1, \dots, -1}_{q})
  \label{CLdefRel}
\end{equation}
The $Cl(p,q)$ is a $2^N$ dimensional algebra whose elements are given by a product of the generators:
\[
 1, \ \gamma_i, \ \gamma_i \gamma_j, \ \gamma_i \gamma_j \gamma_k,\ \dots, \ \gamma_1 \gamma_2 \cdots \gamma_N 
\]
\begin{lem} \label{LEM:Grass}
The $\ZZ$ graded Grassmann numbers 
$ \zeta_{\bm{\alpha}, i} $ are realized in terms of the ordinary Grassmann number $\xi_{\mu} $ and the Clifford algebra $ Cl(p,q),\; p+q = 2:$
\begin{align*}
  \zeta_{(0,0),m} &= 1 \otimes x_m, \quad 
  \zeta_{(1,0),\mu} = \gamma_1 \otimes \xi_{\mu},
  \\
  \zeta_{(0,1),\mu} &= \gamma_2 \otimes \xi_{\mu},  \quad
  \zeta_{(1,1),m} = \gamma_1 \gamma_2 \otimes x_m,
\end{align*}
where $ x_m \in \mathbb{R}. $ 
\end{lem}

%
\section{Verma modules and their reducibility.} \label{SEC:VMVM}

 We want to investigate the lowest weight representations of $\g$ by employing the standard procedure of Lie theory \cite{Dix}. 
However, it turns out that the naive approach is not successful in defining Verma modules over $\g.$ 
We begin with the failure of the naive construction of Verma modules and then consider a modified approach.  

\subsection{Failure of the naive construction.} \label{Failure}

The algebra $\g$ has a natural triangular decomposition: 
\begin{align}
   \g_+ &= \langle \; A_+, \ a_+, \ b_+ \; \rangle,
   \nonumber \\
   \g_0 &= \langle \; N, \ F \; \rangle,
   \nonumber \\
   \g_- &= \langle \; A_-, \ a_-, \ b_- \; \rangle. \label{TriDeco1}
\end{align}
This is based on the eigenvalues of ad$N$ which are given as follows:
\begin{equation*}
  \begin{array}{ccccc} 
   +2 & +1 & 0 & -1 & -2 \\\hline
   A_+ & a_+,\; b_+ & N, \; F & a_-,\; b_- & A_-
  \end{array}
\end{equation*}
The anti-involution $ \omega $ \eqref{antiinv} acts on the subspaces as
\begin{equation}
   \omega(\g_{\pm}) = \g_{\mp}, \qquad \omega(\g_0) = \g_0
\end{equation}
and it is easy to see that 
$ \llbracket \g_0, \g_{\pm} \rrbracket \subseteq \g_{\pm} $ and
$ \llbracket \g_{\pm}, \g_{\pm} \rrbracket \subseteq \g_{\pm}. $ 
Therefore, \eqref{TriDeco1} is a natural triangular decomposition of $\g.$  

The decomposition \eqref{TriDeco1} leads us to define Verma modules in a standard way. 
We first define the lowest weight state $\ket{h,\varphi}$ by
\begin{equation}
     a_- \ket{h,\varphi} = 0, \qquad N \ket{h,\varphi} = h \ket{h,\varphi}, \qquad F \ket{h,\varphi} = \varphi \ket{h,\varphi}.
\end{equation}
It follows that
\begin{equation}
    A_- \ket{h,\varphi} = b_- \ket{h,\varphi} = 0.
\end{equation}
Then the Verma modules over $\g$ are defined by 
$ M(h,\varphi) = {\mathcal U}(\g_+) \otimes \ket{h,\varphi} $ where 
$ {\mathcal U}(\g_+) $ is the universal enveloping algebra of $\g_+.$ 
The basis of $ M(h,\varphi) $ is obviously given by
\begin{equation}
   \ket{k,\ell} = (a_+)^k (b_+)^{\ell} \ket{h,\varphi}, \quad k, \ell \in \mathbb{Z}_{\geq 0}
\end{equation}
It is immediate to see that the relations 
$ \{ a_+, a_+ \}= \{b_+, b_+\} = 2 A_+ $ are not realized on this basis. 
This shows the failure of the naive construction of Verma modules.

\subsection{Verma modules and singular vectors.} \label{SEC:VERMA}

To overcome the difficulty in \S \ref{Failure} we use a $ \ZZ$ graded Grassmann number $ \zeta $ of degree $(1,1).$ Suppose further that $ \zeta^2 = 1. $ 
Legitimacy of the assumption is ensured by Lemma \ref{LEM:Grass} since $ \zeta $ may be realized by using $Cl(1,1).$ 
Now we define a new basis of $\g$:
\begin{equation}
   c_{\pm}= \frac{1}{\sqrt{2}} (a_- \pm  b_-\zeta), \qquad d_{\pm} = \frac{1}{\sqrt{2}} (a_+ \pm \zeta b_+), \qquad \tF = \hf \zeta F. 
   \label{NewBasis}
\end{equation}
The non-vanishing defining relations for the new basis are written as
\begin{align}
    \{ c_+, c_-\} &= 2A_-, & \{d_+, d_-\} &= 2A_+,  & \{ c_{\pm}, d_{\pm} \} &= 2(N\mp \tF), 
    \nonumber \\
     [\tF, c_{\pm} ] &= \mp c_{\pm}, & [\tF,d_{\pm}] &= \pm d_{\pm}.
\end{align}
We remark that $ c_{\pm} $ and $ d_{\pm} $ are nilpotent and that $ c_{\pm} $ anticommutes with $d_{\mp}$. 
The anti-involution $ \omega $ \eqref{antiinv} is extended to the new basis \eqref{NewBasis} by setting $ \omega(\zeta) = 1:$ 
\begin{equation}
  \omega( c_{\pm} ) = d_{\pm}, \qquad \omega(\tF) = \tF. 
\end{equation}

In fact, this change of basis converts the $\ZZ$ graded superalgebra to an ordinary superalgebra of $\mathbb{Z}_2$ grading. 
Due to the degree of $\zeta$, the degree of $ c_{\pm}$ and $ d_{\pm} $ are all $(0,1)$, while the degree of $\tF$ is $ (0,0). $ Thus one may consider representations of the $\mathbb{Z}_2$ graded algebra, then convert it to the ones for $ \ZZ $ grading.

The present choice of the basis of $\g$ diagonalize  ad$N$ and ad$\tF.$ 
Their eigenvalues are summarized as
\begin{equation}
  \begin{array}{llll}
     A_+ \; (+2,0), & A_- \; (-2,0), & c_+\; (-1,-1), & c_-\;(-1,+1), \\[3pt]
     N \; (0,0),    & \tF\; (0,0),   & d_+\; (+1,+1), & d_-\;(+1,-1). 
  \end{array}
\end{equation}
We introduce the triangular decomposition of $\g$ according to the eigenvalue of ad$N:$
\begin{equation}
   \g_+ = \langle \;  A_+, \ d_{\pm} \; \rangle, 
   \qquad
   \g_0 = \langle \; N, \ \tF \; \rangle, 
   \qquad
   \g_- = \langle \; A_-, \ c_{\pm} \; \rangle. \label{TriColor}
\end{equation}
Define the lowest weight vector $\ket{h,f}$ by
\begin{equation}
   c_{\pm} \ket{h,f} = 0, \qquad N \ket{h,f} = h \ket{h,f}, \qquad \tF \ket{h,f} = f \ket{h,f}.
\end{equation}
It then follows that $ A_- \ket{h,f} = 0.$ 
We define the Verma modules over $\g$ by a space induced from $\ket{h,f}$ by   
$ M(h,f) = {\mathcal U}(\g_+) \otimes \ket{h,f}. $ 
The natural basis of $ M(h,f) $ is given by
\begin{equation}
  \ket{k,\mu,\nu} = (A_+)^k (d_+)^{\mu} (d_-)^{\nu} \ket{h,f}, \quad k \in \mathbb{Z}_{\geq 0}, \ \mu, \nu \in \{0, 1\} 
  \label{VMbasis}
\end{equation}
It is then straightforward to compute the action of $ \g $ on $ M(h,f). $ 
The action of $ \g_0 $ yields
\begin{align}
   N \ket{k,\mu,\nu} &= (h+2k+\mu+\nu) \ket{k,\mu,\nu},
   \nonumber \\
   \tF \ket{k,\mu,\nu} &= (f+\mu-\nu) \ket{k,\mu,\nu}.
\end{align}
The action of $\g_+$ is given by 
\begin{align}
  A_+ \ket{k,\mu,\nu} &= \ket{k+1,\mu,\nu},
  \nonumber \\
  d_+ \ket{k,\mu,\nu} &= \delta_{\mu,0} \ket{k,\mu+1,\nu},
  \nonumber \\
  d_- \ket{k,\mu,\nu} &= (-1)^{\mu} \delta_{\nu,0} \ket{k,\mu,\nu+1} + \delta_{\mu,1} 2\ket{k+1,\mu-1,\nu}.
\end{align}
The action of $ \g_-$ is summarized as
\begin{align}
   A_- \ket{k,\mu,\nu} &= 4k(h+k+\mu+\nu-1) \ket{k-1,\mu,\nu} + \delta_{\mu,1}\delta_{\nu,1} 4(h+f) \ket{k,\mu-1,\nu-1},
   \nonumber \\
   c+ \ket{k,\mu,\nu} &=  \delta_{\mu,1} 2(h+2k+2\nu - f) \ket{k,\mu-1,\nu} + (-1)^{\mu} \delta_{\nu,0}2k\ket{k-1,\mu,\nu+1},
   \nonumber \\
   c_- \ket{k,\mu,\nu} &= (-1)^{\mu} \delta_{\nu,1} 2(h+f) \ket{k,\mu,\nu-1} + \delta_{\mu,0} 2k \ket{k-1,\mu+1,\nu}.
\end{align}
We have successfully obtained the Verma modules over $\g.$ 

Let us now discuss reducibility of the Verma modules. 
This may be done by singular vectors \cite{Dix}. 
The existence of the singular vector, by definition, means that $ M(h,f)$ is reducible. The Verma module $M(h,f)$ has a natural grading as a vector space:
\begin{equation}
  M(h,f) = \bigoplus_{m \in \mathbb{N}} M_m(h,f), \quad 
  M_m(h,f) = \{ \ \ket{v} \in M(h,f) \; | \; N \ket{v} = (h+m) \ket{v} \ \}
\end{equation}
where $\mathbb{N}$ is the set of non-negative integers. 
In the basis \eqref{VMbasis}, the integer $m$ is given by $ m = 2k+\mu+\nu $ and we call $m$ the \textit{level.} Any singular vector is an eigenvector of $N$ so that it must be an element of a subspace $M_m(h,f).$ 
We give a complete list of singular vectors in $M(h,f)$. 
\begin{thm} \label{THE:SV}
 $M(h,f)$ has precisely one singular vector if $h, f$ satisfy one of the following conditions:
 \begin{enumerate}
   \renewcommand{\labelenumi}{(\arabic{enumi})}
   \item $h=f:$ The singular vector is  $ \ket{0,1,0}.$ 
   \item $h=-f:$ The singular vector is  $ \ket{0,0,1}. $  
   \item $h=-n $ and $ f \neq n $ for a positive integer $n:$ 
   The singular vector exists at level $2n$ subspace and given by
   \begin{equation}
       \ket{n,0,0}+ \frac{n}{f-n} \ket{n-1,1,1}.  \label{SVformula}
\end{equation}    
 \end{enumerate}
\end{thm}

\begin{proof}
First of all, we note that $\dim M_m(h,f)=2 $ for any  level $m$. If the level $m=2n+1$ is odd, then $M_m(h,f)$ is spanned by $ \ket{n,0,1}$ and $ \ket{n,1,0}, $ while if $m=2n$, then $M_m(h,f)$ is spanned by $ \ket{n,0,0}$ and $ \ket{n-1,1,1}. $ 
We study the odd and even level separately. 

\medskip
\noindent
(i) $m=2n+1.$ Since $\ket{n,0,1} $ and $ \ket{n,1,0}$ have distinct eigenvalues of $ \tF,$ singular vector is not a linear combination of the two vectors. Requirement that $\ket{n,0,1}$ is annihilated by the action of $\g_-$ provides the relations
\begin{equation}
   n=0, \qquad h+f = 0.
\end{equation}
The same argument for $ \ket{n,1,0} $ provides the relations
\begin{equation}
  n=0, \qquad h-f = 0.
\end{equation}
This proves (1) and (2) of the theorem. 

\medskip\noindent
(ii) $m=2n.$ A singular vector, if any, may be written as
\begin{equation}
  \ket{u} = \ket{n,0,0} + \alpha \ket{n-1,1,1}
\end{equation}
with a constant $ \alpha. $ $\ket{u}$ is an eigenvector of $\tF$ with the eigenvalue $f.$ 
The requirement that the action of $\g_-$ annihilate $\ket{u}$ gives the relations
\begin{align}
  & n+ \alpha (h+2n-f) = 0, \nonumber \\
  & n - \alpha (h+f) = 0, \nonumber \\
  & n(h+n-1) + \alpha (h+f) = 0,\nonumber \\
  & \alpha (n-1) (h+n) = 0. \label{CondSV}
\end{align}
The last equation implies that there exists three possibilities:
\[
   \text{(a)}\ \alpha = 0, \qquad \text{(b)} \ h = -n, \qquad \text{(c)} \ n=1.
\]
The case (a) means that $n = 0 $ so that $ \ket{u} = \ket{0,0,0}. $ Thus this case is trivial. 
The case (b), the first three equations in \eqref{CondSV} are reduced to single relation:
\begin{equation}
   n + \alpha (n-f) =0. \label{alpha-eq}
\end{equation} 
If $f=n,$ then the relations means that $n=0.$ Thus we need $ \alpha = 0$ and we have a only trivial solution. 
If $ f \neq n,$ then \eqref{alpha-eq} is solved to $\alpha $ and give the unique singular vector \eqref{SVformula}. 
The case (c) is reduced to a subclass of the case (b) since we find the relation $ h=-1 = -n.$  
We thus completed the proof. 
\end{proof}

\begin{cor}
 The Verma module $M(h,f)$ is irreducible if following conditions are true:
 \begin{enumerate}
   \renewcommand{\labelenumi}{(\arabic{enumi})}
   \item $ h \neq \pm f$
   \item $h \neq -n $ or $ f = n$
 \end{enumerate}
\end{cor}

%
\section{Invariant partial differential equations.} \label{SEC:PDE}

In this section, 
we derive partial differential equations which are invariant under the color supergroup generated by $\g. $ 
This is done by employing the method of \cite{Dob}. 
We give a brief outline of the method and refer to \cite{Dob} for detail. 
The basic idea of the method is to realize the Verma modules, discussed in \S \ref{SEC:VERMA}, by a space of differentiable functions defined on the color supergroup. 

Let $ \cal G$ be the color supergroup generated by $\g$ which means that an element $ g \in {\cal G} $ is written as $ g = \exp(X), \ X \in \g. $ 
We consider a space of functions $ C_{\Lambda}$ on $\cal G$ having a special property. 
That is, $ F \in C_{\Lambda}$ satisfies the relation:
\begin{equation}
  F(g g_0 g_-) = \exp(\Lambda(H)) F(g), \quad g \in {\cal G}, \ g_0 = \exp(H), \ H \in \g_0, \ g_- \in \exp(\g_-)
  \label{RightCov}
\end{equation}
where $ \Lambda(H) $ is an eigenvalue of $H. $ 
This property means that  $F$ is, in fact, a function on $ {\cal G}_+ = \exp(\g_+). $ 
We then define the left $(\pi_L)$ and right $(\pi_R)$  actions of $\g $ on $ C_{\Lambda}$ according to the standard way of Lie theory:
\begin{align}
   \pi_L(X) F(g) &= \left. \frac{d}{d\tau} F(e^{-\tau X}g) \right|_{\tau=0}, \label{Laction}\\
   \pi_R(X) F(g) &= \left. \frac{d}{d\tau} F(ge^{\tau X}) \right|_{\tau=0},  \label{Raction}
\end{align}
where $ X \in \g. $ Note that the parameter $\tau $ is a $ \ZZ $ graded Grassmann number of degree same as $X. $ 
A definition of derivative with respect to the $ \ZZ $ graded Grassmann numbers is found in \cite{AiSe}. 
It is then easy to see that, due to the property \eqref{RightCov}, the function $ F \in C_{\Lambda}$ plays the role of a lowest weight vector $ \ket{h,f} $ by the right action.  Namely, $ \pi_R(X) F = 0 $ if $ X \in \g_-$ and $ \pi_R(X) F = \Lambda(x) F $ if $ X \in \g_0.$  While $ \pi_R(X) $ with $ X \in \g_+ $ becomes a differential operator so that a Verma module $M(h,f)$ is realized by the action of $ \pi_R(X) $ with $ X \in \g_+ $ on the function $F(g).$ 

To be more explicit, we parametrize $ g_+ \in {\cal G}_+ $ as
\begin{equation}
  g_+ = \exp(x A_+) \exp(\psi_+ d_+) \exp(\psi_- d_-). 
\end{equation}
Then we obtain from \eqref{Raction}
\begin{align}
   \pi_R(A_+) &= \frac{\partial}{\partial x}, \nonumber \\
   \pi_R(d_+) &= \frac{\partial}{\partial \psi_+} + 2 \psi_- \frac{\partial}{\partial x}, \nonumber \\
   \pi_R(d_-) &= \frac{\partial}{\partial \psi_-}
\end{align}
and by setting $ \Lambda(N) = h, \ \Lambda(\tF)=f$ we have 
\begin{equation}
  \pi_R(N) = h, \qquad \pi_R(\tF) = f.  
\end{equation}
As shown in \S \ref{SEC:VERMA} the singular vectors in $M(h,f)$ is written as $ {\cal P} \ket{h,f} $ with $ {\cal P} \in U(\g_+). $ 
It is shown in \cite{Dob} that differential equations invariant under the color supergroup $\cal G $ are given by $ \pi_R({\cal P}) \varphi = 0. $ The symmetry transformations are generated by the left action \eqref{Laction}. 
This is due to the fact that a singular vector is an intertwining operator of two representations $M(h,f)$ and $M(h',f'). $ 
According to Theorem  \ref{THE:SV} we found a hierarchy of invariant differential equations.
\begin{prop}
 The following equations are invariant under the color supergroup $\cal G:$
 \begin{align}
     &  \frac{\partial}{\partial \psi_-} \varphi(x,\psi_{\pm}) = 0, \nonumber \\
     & \left( \frac{\partial}{\partial \psi_+} + 2 \psi_- \frac{\partial}{\partial x} \right)  \varphi(x,\psi_{\pm}) = 0, \nonumber \\
     & \left[ \frac{\partial}{\partial x} + \frac{n}{f-n} \left( \frac{\partial}{\partial \psi_+} + 2\psi_- \frac{\partial}{\partial x} \right) \frac{\partial}{\partial \psi_-} \right] \left( \frac{\partial}{\partial x} \right)^{n-1}  \varphi(x,\psi_{\pm}) = 0. 
 \end{align}
\end{prop}
 
Up to here our analysis of the representations of $\g$ is done by converting the basis of $\g$ into a $\mathbb{Z}_2$ grading ones. 
Thus the invariant equations are written in terms of the variables of degree $(0,0)$ and $(0,1). $ 
We now rewrite the invariant differential equations in a form also containing the variables of degree $(1,0)$ and $(1,1).$ 
Let $\psi $ and $ \theta $ be $\ZZ$ graded Grassmann number of degree $(0,1)$ and $ (1,0),$ respectively.  
Taking the basis change \eqref{NewBasis} into account we set 
\begin{equation}
 \psi_{\pm} = \frac{1}{\sqrt{2}} ( \psi \pm \zeta \theta).
\end{equation} 
and replace $ \psi_{\pm} $ with $\psi $ and $ \theta. $ 
Then the independent variables of the invariant differential equations become $ x, \psi $ and $ \theta$ and the equations yield as follows:  
\begin{align*}
   & \left( \frac{\partial}{\partial \psi} + \zeta \frac{\partial}{\partial \theta} \right) \varphi= 0,  \\
   & \left[ \frac{\partial}{\partial \psi} - \zeta \frac{\partial}{\partial \theta} +2(\psi - \zeta \theta) \frac{\partial}{\partial x} \right] \varphi=0, \\\
   &  \left[ \frac{\partial}{\partial x} + \frac{n}{f-n} \left(  \psi \frac{\partial}{\partial \psi} + \theta \frac{\partial}{\partial \theta}  - \zeta \Big( \psi \frac{\partial}{\partial \theta}  + \theta \frac{\partial}{\partial \psi} \Big) \right) \frac{\partial}{\partial x} - \frac{n \zeta}{f-n}  \frac{\partial^2}{\partial \psi \partial \theta} \right]
   \left( \frac{\partial}{\partial x} \right)^{n-1}  \varphi = 0.
\end{align*}

%
\section{Concluding remarks.}  \label{SEC:CR}

We studied lowest weight representations of the $\ZZ$ graded superalgebra given by Rittenberg and Wyler. 
Reducibility of the Verma modules over the  $\ZZ$ graded superalgebra is shown by explicit construction of singular vectors. 
As an application of the present scheme invariant partial differential equations defined on the space of functions whose variables are $\ZZ$ graded Grassmann numbers are obtained. 
The present scheme may apply to other color superalgebras and one may find many related differential equations. 
Since differential equations are very basic objects in theoretical physics and many areas of mathematics, 
the present line of investigation will reveal a connection of  color superalgebras with different areas of science. 

We close this paper with two possible lines of further investigations on applications of color superalgebras. 
The first one is mathematical application. It is well known that representations of Lie algebras and Lie groups are closely related to orthogonal polynomials. We anticipate that representations of color algebras or groups also have such connection. 
The second one is physical application. One may use a nonlinear realization method to write down Lagrangians which are invariant under a color supergroup, since the method is originally defined for Lie groups then extended to supergroups. 
Color supergroups are natural generalization of supergroups, thus nonlinear realization methods will be generalized to any color supergroups. If this is done, one may find a dynamical system whose symmetry is governed by a color supergroup. 
This would open the way to physical applications of color supergroups.


\subsection*{Acknowledgment.}

The author is supported by the  grants-in-aid from JSPS (Contract No. 26400209).

%
%

\vspace{1cm}
\noindent
Department of Physical Science \\
Graduate School of Science \\
Osaka Prefecture University \\
Nakamozu Campus, Sakai, Osaka 599-8531 \\
Japan

\end{document}